\documentclass[sigconf]{acmart}

\newcommand{\myparatight}[1]{\smallskip\noindent{\bf {#1}:}~}

\usepackage{amsthm}
\usepackage{amsmath}

\usepackage{amssymb}

\usepackage{amsthm}
\usepackage{amsmath}

\usepackage{float}
\usepackage{graphics}
\usepackage{graphicx}
\usepackage[skip=0pt]{caption}
\usepackage{algorithm}
\usepackage{algpseudocode}
\usepackage{bm}
\usepackage{color}
\usepackage{multirow}
\usepackage{makecell}
\usepackage{subfig}
\usepackage{soul}

\usepackage{wrapfig}

\allowdisplaybreaks

\usepackage{xspace}

\newcommand{\alg}{$\mathsf{BRACE}~$}
\newcommand{\algns}{$\mathsf{BRACE}$}

\usepackage{hyperref}

\usepackage{stfloats}

\newtheorem{assumption}{Assumption}
\newtheorem{thm}{Theorem}

\newtheorem*{remark}{Remark}

\copyrightyear{2025} 
\acmYear{2025} 
\setcopyright{cc}
\setcctype{by}
\acmConference[WWW Companion '25]{Companion Proceedings of the ACM Web Conference 2025}{April 28-May 2, 2025}{Sydney, NSW, Australia}
\acmBooktitle{Companion Proceedings of the ACM Web Conference 2025 (WWW Companion '25), April 28-May 2, 2025, Sydney, NSW, Australia}
\acmDOI{10.1145/3701716.3715491}
\acmISBN{979-8-4007-1331-6/25/04}

\makeatletter
\gdef\@copyrightpermission{
  \begin{minipage}{0.3\columnwidth}
   \href{https://creativecommons.org/licenses/by/4.0/}{\includegraphics[width=0.90\textwidth]{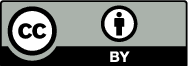}}
  \end{minipage}\hfill
  \begin{minipage}{0.7\columnwidth}
   \href{https://creativecommons.org/licenses/by/4.0/}{This work is licensed under a Creative Commons Attribution International 4.0 License.}
  \end{minipage}
  \vspace{5pt}
}
\makeatother

\begin{document}

\title{Byzantine-Robust Federated Learning over Ring-All-Reduce Distributed Computing}

\author{Minghong Fang}
\affiliation{
	\institution{University of Louisville}
	\city{Louisville}
        % \state{KY}
         \country{USA}
}

\author{Zhuqing Liu}
\affiliation{
	\institution{University of North Texas}
	\city{Denton}
        % \state{TX}
        \country{USA}
}

\author{Xuecen Zhao}
\authornote{Xuecen Zhao conducted this research while she was an intern under the supervision of Minghong Fang.}
\affiliation{
	\institution{University of Washington}
	\city{Seattle}
        % \state{WA}
        \country{USA}
}

\author{Jia Liu}
\affiliation{
	\institution{The Ohio State University}
	\city{Columbus}
        % \state{OH}
        \country{USA}
}

\begin{abstract}
Federated learning (FL) has gained attention as a distributed learning paradigm for its data privacy benefits and accelerated convergence through parallel computation. Traditional FL relies on a server-client (SC) architecture, where a central server coordinates multiple clients to train a global model, but this approach faces scalability challenges due to server communication bottlenecks. To overcome this, the ring-all-reduce (RAR) architecture has been introduced, eliminating the central server and achieving bandwidth optimality. However, the tightly coupled nature of RAR's ring topology exposes it to unique Byzantine attack risks not present in SC-based FL. Despite its potential, designing Byzantine-robust RAR-based FL algorithms remains an open problem. To address this gap, we propose \alg (\ul{B}yzantine-\ul{r}obust ring-\ul{a}ll-redu\ul{ce}), the first RAR-based FL algorithm to achieve both Byzantine robustness and communication efficiency. We provide theoretical guarantees for the convergence of \alg under Byzantine attacks, demonstrate its bandwidth efficiency, and validate its practical effectiveness through experiments. Our work offers a foundational understanding of Byzantine-robust RAR-based FL design.
\end{abstract}




\maketitle

% !TEX root = mainfile.tex

\section{Introduction} \label{sec:intro}

Federated learning (FL)~\cite{mcmahan2017communication} is a distributed learning paradigm where multiple clients collaboratively train a global model under the coordination of a central server. 
This parallel computation across many clients accelerates the convergence of the training process. However, the server-client (SC) architecture used in FL presents two key limitations. 
% First, the central server is vulnerable to single-point-of-failure risks, making it a target for Byzantine attacks~\cite{yin2018byzantine,blanchard2017machine,fang2020local,shejwalkar2021manipulating,bagdasaryan2020backdoor,zhang2024poisoning} where malicious clients can manipulate the global model via data or model poisoning. 
First, the central server faces single-point-of-failure risks~\cite{fang2024byzantine} and is susceptible to Byzantine attacks~\cite{yin2018byzantine,blanchard2017machine,fang2020local,shejwalkar2021manipulating,bagdasaryan2020backdoor,zhang2024poisoning}, where malicious clients can compromise the global model through data or model poisoning.
While Byzantine-robust algorithms~\cite{blanchard2017machine,yin2018byzantine,bernstein2018signsgd,ozdayi2021defending,cao2020fltrust,fang2024byzantine,fang2022aflguard,fang2025FoundationFL} have been proposed to mitigate such risks, their effectiveness varies. Second, the SC framework suffers from a communication bottleneck, as the server must exchange model parameters and updates with all clients. 
This causes data traffic to scale linearly with clients, limiting FL scalability.

To address the communication scalability challenge in FL, the ``ring-all-reduce'' (RAR)~\cite{patarasuk2009bandwidth,sergeev2018horovod} architecture has been introduced. In this model, clients form a ring topology to share and aggregate updates, eliminating the need for a central server and reducing the communication bottleneck. The key advantage of RAR is that the amount of data exchanged by each client is independent of the number of clients, making it bandwidth optimal. This feature has made RAR increasingly popular in FL systems, supported by major frameworks like TensorFlow and PyTorch, and accelerating tasks such as ImageNet classification and BERT pre-training. However, the ring topology introduces new challenges, particularly in terms of Byzantine attacks. Since clients only communicate with their immediate neighbors, they lack full access to other clients' updates, making it difficult to detect malicious updates. 
This limitation makes existing Byzantine-robust algorithms designed for server-client architectures inapplicable to RAR-based systems, highlighting the need for new, specialized Byzantine-robust algorithms for RAR.

\myparatight{Our work}
In this paper, we aim to bridge this critical gap in Byzantine-robust FL algorithm design.
The main contribution of this paper is that we propose a Byzantine-robust RAR-based FL algorithm called \alg (\ul{B}yzantine-\ul{r}obust ring-\ul{a}ll-redu\ul{ce}).
To our knowledge, \alg is the first Byzantine-robust RAR-based FL algorithm in the literature.
\alg retains the core features of the RAR architecture, where clients form a logical ring and only communicate with their immediate neighbors. 
To enhance resilience against Byzantine attacks, \alg introduces four key innovations: 1) clients apply 1-bit quantization (element-wise sign) to local updates, reducing sensitivity to poisoned data or updates; 2) clients exchange sub-vectors of local updates with neighbors, performing reduction operations before forwarding the results; 3) a consensus mechanism ensures that only a majority consensus determines the aggregated sign, further mitigating the impact of poisoning attacks; and 4) upon completing the consensus and exchange process, all clients receive an identical aggregated local update, which is used to update the global model.
The key contributions of this paper are summarized as follows:

\begin{list}{\labelitemi}{\leftmargin=1em \itemindent=-0.08em \itemsep=.2em}

\item 
We introduce the \alg algorithm, the first RAR-based FL method that achieves both robustness and communication efficiency.

\item 
We theoretically prove that \alg achieves an \(\mathcal{O}(1/T)\) convergence rate under Byzantine attacks and retains the RAR architecture's bandwidth-optimal property, where $T$ is the total number of training rounds. Additionally, \alg maintains an \(\mathcal{O}(1)\) per-round communication cost, ensuring scalability.

\item 
We conduct experiments on two real-world datasets, validating that \alg effectively mitigates Byzantine attacks and reduces communication costs in RAR-based FL systems.
\end{list}

\begin{figure*}[t]
	\centering
	\includegraphics[width=0.75\textwidth]{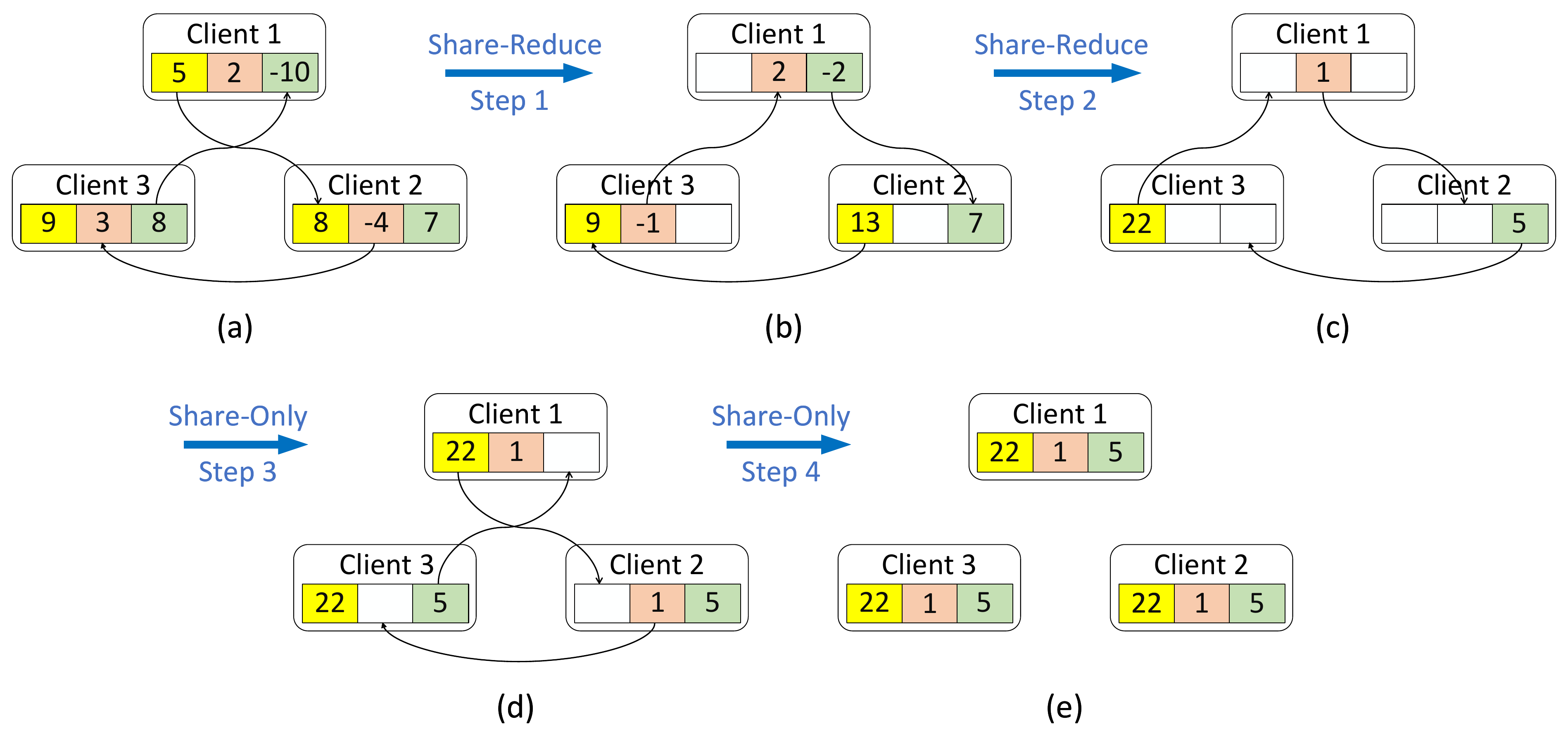}
	\caption{Illustration of the ring-all-reduce (RAR) process.}
	\label{Ring_AllReduce_sys}
	\vspace{-0.10in}
\end{figure*}
% !TEX root = mainfile.tex

\section{Preliminaries}
\label{sec:preliminaries}

\subsection{Federated Learning (FL): An Overview} 
\label{sec:background}

The standard FL system follows a server-client (SC) architecture, where a central parameter server coordinates the training across $n$ clients. Each client $i \in [n]$ has a local dataset, denoted as $\mathcal{D}_i$, and trains a global model collaboratively with the other clients, where $[n]$ denotes the set $\left\{ 1, 2,...,n\right\}$. 
The primary objective of a FL system is to address the optimization problem
$
\arg\min_{\bm{w} \in \mathbb{R}^d} f(\bm{w}) = \frac{1}{n}\sum_{i\in [n]}f_i(\bm{w}),
$
where $d$ is the dimension of $\bm{w}$, $f_i$ is the loss function on client $i$.
FL tackles this problem by distributing the computation across clients.
In each training round $t$, the server broadcasts the current model parameter $\bm{w}^t$ to all clients. Each client $i$ computes a local stochastic gradient $\bm{g}_i^t$ using its dataset and the received model, then sends it back to the server. The server aggregates the gradients using an aggregation rule $\mathsf{GAR(\cdot)}$ and updates the model parameter as
$
\bm{w}^{t+1} =  \bm{w}^{t} -\eta \cdot \mathsf{GAR}(\bm{g}_1^t, \bm{g}_2^t,...,\bm{g}_n^t),
$
where $\eta$ is the learning rate.
The standard aggregation rule, known as the ``Mean'' operation, averages all received stochastic gradients as $\mathsf{GAR}(\bm{g}_1^t, \bm{g}_2^t,...,\bm{g}_n^t) = \frac{1}{n} \sum_{i \in [n]} \bm{g}_i^t$.

\subsection{The Ring-All-Reduce (RAR) Architecture} \label{sec:Ring_arch}

Although the SC architecture speeds up training by enabling parallel computation, it suffers from a key limitation: the server becomes a communication bottleneck as the number of clients increases. In the SC setup, the server broadcasts model parameters to all clients, which compute and return their gradients.
Suppose the stochastic gradient size for each dimension is $m$ bits. For a $d$-dimensional model, this results in a communication cost of $mnd$ bits per round, where $n$ is the number of clients, leading to a linear increase in communication as $n$ grows.
This bottleneck hampers scalability.

To overcome this, the ring-all-reduce (RAR) architecture was proposed, where clients form a {\em ring} to exchange and aggregate model parameters without a central server, eliminating the bottleneck. 
Specifically, in an $n$-client RAR system, each client has an upstream and downstream neighbor and communicates with them in a fixed order, splitting its gradient into $n$ roughly equal-sized sub-vectors.
For instance, in a three-client system (illustrated in Fig.~\ref{Ring_AllReduce_sys}(a)) with gradients $\bm{g}_1=[5, 2, -10]^\top$, $\bm{g}_2=[8, -4, 7]^\top$, and $\bm{g}_3=[9, 3, 8]^\top$, RAR employs ``Share-Reduce'' and ``Share-Only'' phases to aggregate and exchange gradients among clients as follow:

\myparatight{Phase I (Share-Reduce)}In the ``Share-Reduce'' phase, each client receives a sub-vector from its upstream neighbor, reduces it by adding it to its local gradient, and sends the result to the downstream client. For example, in Figs.~\ref{Ring_AllReduce_sys}(a)--(b), Client 1 receives the third sub-vector from Client 3 and sends the reduced result of $-2$ to Client 2. 
After \(n-1\) steps, each client possesses a sub-vector of the final aggregated gradient. As illustrated in Fig.~\ref{Ring_AllReduce_sys}(c), Client 1 holds the fully aggregated gradient for the second sub-vector. Specifically, in the pink block, the computation yields \(2 - 4 + 3 = 1\).

\myparatight{Phase II (Share-Only)}The second phase, ``Share-Only'', involves clients sharing data in a ring without performing gradient reduction. After $n-1$ steps, each client obtains the fully aggregated stochastic gradient, $\sum_{i \in [n]} \bm{g}_i^t$, as shown in Fig.~\ref{Ring_AllReduce_sys}(e). Once the RAR process is complete, each client uses the aggregated gradient to update the model parameter {\em individually} as
$
\label{global_sgd_ring}
\bm{w}_i^{t+1} =  \bm{w}_i^{t} - \frac{\eta}{n} \sum\nolimits_{i \in [n]} \bm{g}_i^t,
$
where $\bm{w}_i^t$ is the model parameter of client $i \in [n]$ in the round $t$.

In the RAR process, in every training round, each client sends $\frac{md}{n}$ bits of data for $2(n-1)$ steps, resulting in a total of $\frac{2md(n-1)}{n}$ bits sent. This amount of data is asymptotically independent of the number of clients $n$ as $n$ grows, meaning there is no communication bottleneck even as the client count increases. This ``bandwidth-optimal'' feature \cite{patarasuk2009bandwidth} has made the RAR architecture increasingly popular in the FL community.

% !TEX root = mainfile.tex

\section{The \alg Algorithm} 
\label{sec:alg}

\begin{figure*}[t!]
	\centering
	\includegraphics[width=0.7\textwidth]{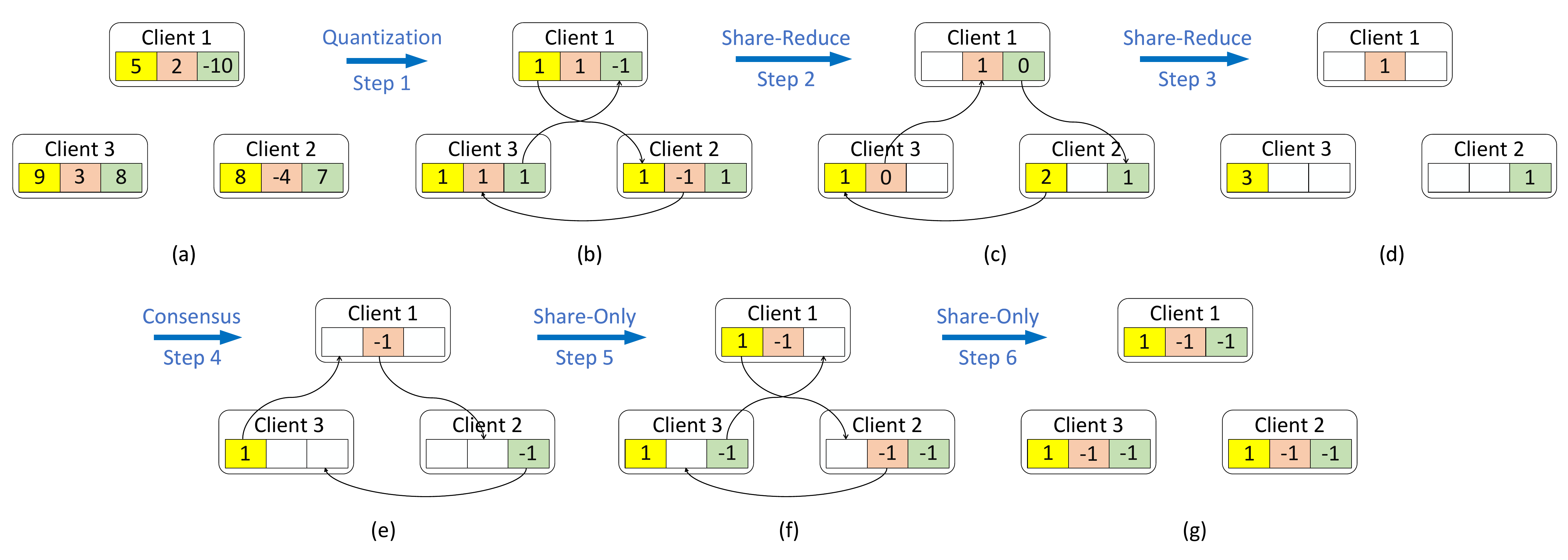}
\caption{An example of the \alg algorithm.}
	\label{Ring_AllReduce_sys_gurad}
	\vspace{-.10in}
\end{figure*}

RAR-based FL systems, though ``bandwidth-optimal'', face unique Byzantine attack risks absent in traditional SC-based FL. The ring topology in RAR, lacking a centralized server, creates a strong coupling among clients, restricting communication to upstream and downstream neighbors. Each client only has a ``partial view'' of the global model, both when receiving and sending sub-vectors of local model updates. This limitation prevents access to the full updates of other clients, making it difficult to distinguish between malicious and benign updates, thereby increasing vulnerability.
For example, if the attacker compromises Client 1 and alters its gradient from $[5, 2, -10]^\top$ to $[5, 2, -200]^\top$ in Fig.~\ref{Ring_AllReduce_sys}(a), the final aggregated gradient becomes $[22, 1, -185]^\top$.
Note that this paper focuses on attacks in which clients modify gradients (either through data or model poisoning) before the RAR process begins, while still faithfully performing the RAR communication steps. This can be ensured using an integrity assurance framework~\cite{roth2022nvidia} to safeguard the execution of FL tasks.
In this paper, we propose \algns, a Byzantine-robust RAR-based FL algorithm where clients form a ring and communicate only with their neighbors. Each client splits its gradient into $n$ sub-vectors, and \alg performs four phases for gradient exchange and aggregation in each training round.

\myparatight{Phase I (Quantization)}Each client applies quantization by taking the element-wise sign of the stochastic gradient (see Step 1 in Fig.~\ref{Ring_AllReduce_sys_gurad}). This process serves two purposes: 1) it reduces sensitivity to malicious gradients from Byzantine attacks, and 2) it lowers the number of bits exchanged during the RAR communication process.

\myparatight{Phase II (Share-Reduce)}In this phase, each client follows the same procedure as in the ``Share-Reduce'' phase of traditional RAR. Specifically, each client obtains a quantized gradient sub-vector from its upstream neighbor, combines it with its local gradient sub-vector through reduction (such as addition), and then forwards the updated sub-vector to its downstream neighbor.
By the conclusion of the Share-Reduce phase, each client retains a portion of the fully aggregated, quantized sub-vector gradient, see Fig.~\ref{Ring_AllReduce_sys_gurad}(d).

\myparatight{Phase III (Consensus)}In this phase, clients apply an additional quantization step to the aggregated quantized sub-vector gradient. 
As \alg operates on each dimension independently. Thus for the dimension $k \in [d]$, we perform the following consensus mapping:
\begin{align}
\mathsf{GAR}(\bm{g}_1^t, \bm{g}_2^t,...,\bm{g}_n^t)[k] = 
\begin{cases}
1, & \text{if } \sum_{i \in [n]} \text{sign}(\bm{g}_i^t[k]) > \lambda, \\
-1, & \text{if } \sum_{i \in [n]} \text{sign}(\bm{g}_i^t[k]) \leq \lambda,
\end{cases} 
\end{align}
where \(\bm{g}_i^t[k]\) represents the \(k\)-th dimension of the gradient \(\bm{g}_i^t\), and \(\mathsf{GAR}(\bm{g}_1^t, \bm{g}_2^t, \ldots, \bm{g}_n^t)[k]\) denotes the \(k\)-th dimension of the aggregated gradient. 
\(\lambda\) is a threshold to determine how much agreement among client gradients is necessary for the aggregated gradient in a specific dimension to be positive or negative.
As an illustration, setting \(\lambda = 2\) yields the mapping results depicted in Fig.~\ref{Ring_AllReduce_sys_gurad}(e).

\myparatight{Phase IV (Share-Only)}The final phase is the ``Share-Only'' phase, where each client performs the same steps as in the ``Share-Only'' phase of standard RAR. After $n-1$ steps, all clients receive the same fully aggregated quantized gradient, as illustrated in Fig.~\ref{Ring_AllReduce_sys_gurad}(g).

Next, we analyze the communication cost of \alg algorithm.  
Because both ``Quantization'' and ``Consensus'' are performed locally, these operations introduce no additional communication overhead. 
During the ``Share-Reduce'' phase, each client transmits \(\frac{md}{n}\) bits per step over \((n-1)\) steps, and \(\frac{d}{n}\) bits per step over \((n-1)\) steps in the ``Share-Only'' phase (note that our \alg compresses gradients by reducing the gradient for each dimension to 1 bit before the ``Share-Only'' phase). This results in a total communication cost of \(\frac{d(n-1)(m+1)}{n}\) bits. Since \(m > 1\), this cost is less than \(\frac{2md(n-1)}{n}\), making the per-round communication of \alg more efficient than the standard RAR process. 
Furthermore, as the number of clients \(n\) increases, \(\frac{d(n-1)(m+1)}{n}\) remains asymptotically constant. This implies that \alg achieves a per-round communication cost of \(\mathcal{O}(1)\), independent of the number of clients, thereby ensuring scalability and eliminating communication bottlenecks.
This confirms that \alg preserves the ``bandwidth-optimal'' property of the RAR architecture.
Finally, we summarize the per-round communication costs of SC, RAR and \alg architectures in Table~\ref{comm_cost}.

\begin{table}[t!]
\label{comm_cost}
	\caption{Per-round communication costs of the SC, RAR, and \alg architectures with $n$ clients.}
	\centering
	\label{comm_cost}
		\scriptsize 
	\begin{tabular}{|c|c|c|c|}  \hline
		& SC  & RAR & \alg \\  \hline
		Cost & $mnd$  & $\frac{2md(n-1)}{n}$  & $\frac{d(n-1)(m+1)}{n}$  \\  \hline
	\end{tabular}
	 \vspace{-.15in}
\end{table}

\section{Formal Security Analysis}
\label{sec:sec_ana}

\begin{assumption}
\label{assumption_1}
Gradient is Lipschitz continuous for $L>0$, i.e.,
\begin{align}
\|\nabla f(x) - \nabla f(y)\| \leq L \|x - y\|, \quad \forall x, y \in \mathbb{R}^d. \nonumber
\end{align}
\end{assumption}

\begin{thm}
\label{theorem_1}
If \(0 \leq Pr\left( I_{\left(\sum_{i \in [n]} \text{sign}\left(\bm{g}_i^t[k]\right) > \lambda\right)} = 0 \mid \mathcal{F}_t \right) < 0.5\) for \(i \in [n]\) and \(k \in [d]\), where \(\mathcal{F}_t\) represents the filtration of all random variables at round \(t\), and \(I\) is an indicator function, then under Assumption~\ref{assumption_1}, \alg achieves the following convergence result:
\begin{align}
\frac{1}{T} \sum_{t=1}^{T} \|\nabla f(\bm{w}^t)\| \leq \frac{f(\bm{w}^1) - f^*}{\eta T} + \frac{L\eta^2}{2},\nonumber
\end{align}
where \(T\) is the total training rounds, \(\bm{w}^1\) is the initial model, and \(f^*\) is the minimum of \(f\), $\| \cdot\|$ denotes the $\ell_2$ norm.
\end{thm}

\begin{proof}
	The proof is relegated to Appendix~\ref{sec:appendix_1}.
\end{proof}
% \vspace{-0.10in}
\begin{remark} 
Theorem~\ref{theorem_1} shows that our \alg converges to a stationary point and attains an \(\mathcal{O}(1/T)\) finite-time convergence rate. 
\end{remark}

% !TEX root = mainfile.tex

\begin{table*}[htb!]
	\centering
	\caption{Results of different defenses. The results of Backdoor attack are in the form of ``test error rate / attack success rate".}
	\centering
	\addtolength{\tabcolsep}{-3.1pt}
	\label{tab:error_rate}
	\scriptsize 
	\subfloat[Fashion-MNIST dataset.]
	{
		\begin{tabular}{|c|c|c|c|c|c|c|c|c|c|}
		\hline
		Attack & \multicolumn{1}{c|}{RAR} & \multicolumn{1}{c|}{Krum}  & \multicolumn{1}{c|}{Median} & \multicolumn{1}{c|}{Trim-mean} & \multicolumn{1}{c|}{signSGD} & \multicolumn{1}{c|}{RLR}& \multicolumn{1}{c|}{\alg}\\
        \hline
        None attack &  0.19 & 0.21   & 0.23 &  0.23 & 0.19  & 0.19  & 0.19 \\
        \hline
        LF attack & 0.19 & 0.24  & 0.25  & 0.29  & 0.24 & 0.24  & 0.19 \\
        \hline
        Gaussian attack & 0.27 &0.22  & 0.23 & 0.25& 0.23& 0.22  &  0.19\\
        \hline
        Krum attack & 0.19 & 0.90  & 0.42 & 0.27  & 0.24 & 0.26  & 0.19\\
        \hline
        Trim attack& 0.33  & 0.22  & 0.34 & 0.35  & 0.31 & 0.36  & 0.19\\
        \hline
        Min-Max attack& 0.54 & 0.22  & 0.25 & 0.25  & 0.27 & 0.24  & 0.19\\
        \hline
        Min-Sum attack& 0.48 & 0.21  & 0.28 & 0.26  & 0.22 & 0.27  & 0.19\\
        \hline
        Backdoor attack & 0.20 / 0.94 &  0.22 / 0.03 & 0.28 / 0.06 &  0.24 / 0.10 & 0.20 / 0.03 &  0.20 / 0.03 & 0.19 / 0.03 \\
        \hline
        Adaptive attack & 0.61 & 0.24  & 0.26 & 0.24  & 0.24 & 0.20  & 0.19\\
        \hline
		\end{tabular}%
	}
      \quad
	\vspace{-0.02in}
       \subfloat[CIFAR-10 dataset.]
	{
		\begin{tabular}{|c|c|c|c|c|c|c|c|c|c|}
			\hline
		Attack	& \multicolumn{1}{c|}{RAR} & \multicolumn{1}{c|}{Krum}  & \multicolumn{1}{c|}{Median} & \multicolumn{1}{c|}{Trim-mean} & \multicolumn{1}{c|}{signSGD} & \multicolumn{1}{c|}{RLR}& \multicolumn{1}{c|}{\alg}\\
			\hline
			None attack & 0.24  & 0.26  & 0.26 & 0.25  & 0.24& 0.25  &0.24\\
			\hline
			LF attack &0.31 & 0.26 & 0.32 & 0.29& 0.27& 0.27  & 0.24 \\
			\hline
			Gaussian attack & 0.89 & 0.27  & 0.26 &0.25   & 0.24&  0.26 & 0.24 \\
			\hline
			Krum attack &0.25  & 0.62  &0.29  & 0.31  & 0.26 &0.26   & 0.24\\
			\hline
			Trim attack& 0.42 &  0.29 & 0.44 & 0.53   & 0.38 & 0.41  &0.25 \\
			\hline
                 Min-Max attack& 0.37 &  0.29 & 0.31 &  0.29 & 0.26 &  0.30 & 0.24\\
			\hline
                Min-Sum attack& 0.29 &  0.32 & 0.35 & 0.30 & 0.26 & 0.29  &0.24\\
			\hline
			Backdoor attack & 0.47 / 0.89 & 0.33 / 0.05  & 0.27 / 0.36 &  0.29 / 0.33 & 0.29 / 0.08 &  0.27 / 0.03  & 0.25 / 0.02\\
			\hline
			Adaptive attack & 0.39 & 0.32  & 0.30 & 0.36 &0.28 & 0.31  & 0.26\\
			\hline
		\end{tabular}%
	}
		% \vspace{-0.3in}
        \vspace{-0.2in}
\end{table*}%

\begin{figure*}[!t]
	\centering
	\includegraphics[scale = 0.45]{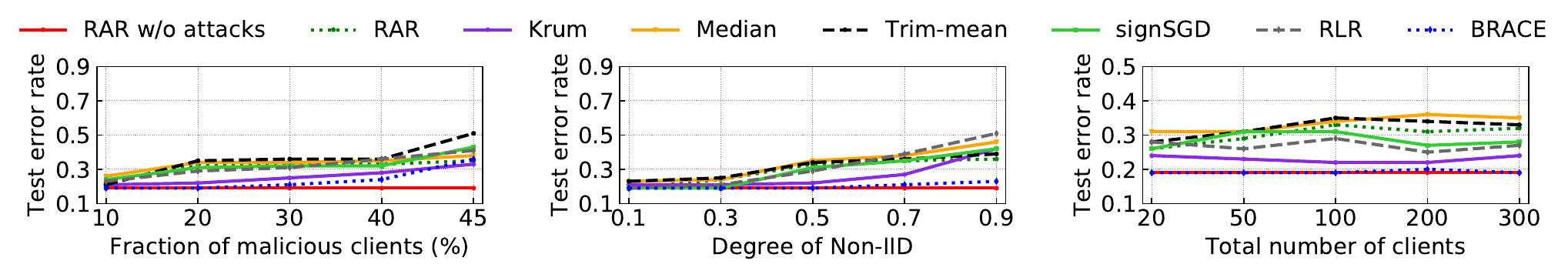}
	\caption{Impact of malicious client fraction, Non-IID degree, and total clients using the Fashion-MNIST dataset.}
\label{attack_size}
\vspace{-0.07in}
\end{figure*}

\begin{figure}[!t]
	\centering
	\includegraphics[scale = 0.45]{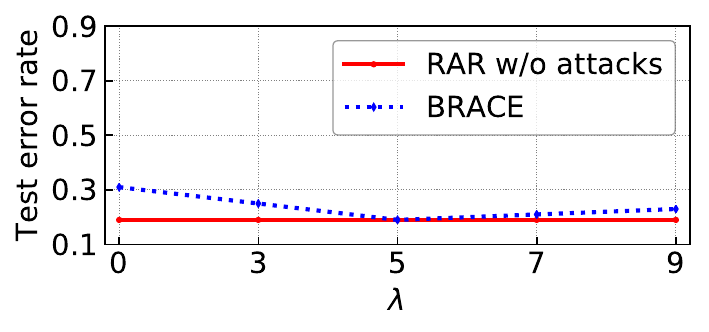}
	\caption{Impact of $\lambda$ using the Fashion-MNIST dataset.}
\label{Effect_of_lambda}
\vspace{-0.22in}
\end{figure}

\vspace{-0.02in}
\section{Experimental Evaluation}  \label{sec:exp}

\subsection{Experimental Setup}
\label{sec:Setup_main}

We compare \alg with six baselines, including the standard Ring-All-Reduce (RAR) method~\cite{patarasuk2009bandwidth,sergeev2018horovod} and five SC-based aggregation rules: Krum~\cite{blanchard2017machine}, Median~\cite{yin2018byzantine}, Trimmed-mean~\cite{yin2018byzantine}, signSGD~\cite{bernstein2018signsgd}, and RLR~\cite{ozdayi2021defending}. Experiments are conducted on two datasets, Fashion-MNIST and CIFAR-10, using \(n=100\) clients, 20\% of which are assumed malicious. We evaluate eight poisoning attacks: Label flipping (LF) attack~\cite{yin2018byzantine}, Gaussian attack~\cite{blanchard2017machine}, Krum attack~\cite{fang2020local}, Trim attack~\cite{fang2020local}, Min-Max attack~\cite{shejwalkar2021manipulating}, Min-Sum attack~\cite{shejwalkar2021manipulating}, Backdoor attack~\cite{bagdasaryan2020backdoor}, and Adaptive attack~\cite{shejwalkar2021manipulating}. 
The models used are a CNN for Fashion-MNIST, with an input of \(28 \times 28 \times 1\), two convolutional layers (30 and 50 filters), max-pooling, a 100-neuron fully connected layer, and a softmax output for 10 classes, and ResNet-20 for CIFAR-10. Training was conducted for 300 rounds on Fashion-MNIST and 500 rounds on CIFAR-10.
To simulate the Non-IID data distribution inherent in FL systems, we follow the method and Non-IID degree is set to 0.5 from~\cite{fang2020local}. 
$\lambda$ is set to 5 in our \algns.
Evaluation metrics include test error rates and attack success rates for Backdoor attack, while only test error rates are used for other attacks. Lower values for these metrics indicate stronger defenses, assuming all clients use the same initial model.

\subsection{Experimental Results} 
%	\vspace{-0.15in}
\label{sec:Results_main}

\myparatight{\alg is effective}Table~\ref{tab:error_rate} summarizes test error rates and attack success rates for defenses on Fashion-MNIST and CIFAR-10.
The ``None attack'' in Table~\ref{tab:error_rate} means all clients are benign.
We make several important observations from Table~\ref{tab:error_rate}.
First, our \alg method is effective in non-adversarial settings.
Secondly, \alg is robust against Byzantine attacks and outperforms existing baselines. In contrast, current methods lack this robustness, leading to higher test error rates and attack success rates under attacks.
Note that, \alg outperforms signSGD by enforcing stronger consensus, ensuring reliable gradient directions in heterogeneous environments. Its thresholding mechanism requires majority consensus to determine the aggregated sign, limiting the influence of malicious clients. In contrast, RLR, which considers the absolute value of the sum and treats positive and negative agreement equally, allows malicious clients to cancel out legitimate contributions, making it more susceptible to manipulation.
Third, \alg approach demonstrates robustness against the Adaptive attack, which performs best when targeting \alg but may not achieve the same effectiveness against other methods due to its tailored design for \algns.

\myparatight{Ablation studies}Fig.~\ref{attack_size} illustrates the impact of malicious client fraction, Non-IID degree, and total clients on different defenses using the Trim-attack and Fashion-MNIST dataset. Our \alg is robust against up to 40\% malicious clients but fails at 45\%, a scenario deemed impractical as per prior studies~\cite{fang2020local,shejwalkar2021manipulating}. Additionally, \alg effectively secures FL systems with highly heterogeneous data, where greater Non-IID degrees indicate more heterogeneity. It also performs well across varying client counts under different poisoning attacks. Fig.~\ref{Effect_of_lambda} highlights the role of $\lambda$ in \algns, showing that small $\lambda$ values allow noisy or adversarial updates to dominate, while excessively large $\lambda$ can exclude legitimate updates, particularly in highly heterogeneous settings.

\myparatight{Communication costs of different methods}As shown in Table~\ref{comm_cost}, the per-round communication cost is \(mnd\) for SC-based aggregation rules like Krum, Median, and Trimmed-mean, and \(nd\) for signSGD and RLR, which use 1-bit gradient quantization (with \(m=1\)). The RAR method incurs a cost of \(\frac{2md(n-1)}{n}\), while our \alg requires \(\frac{d(n-1)(m+1)}{n}\). Here, \(m\) represents the gradient size per dimension in bits, \(n\) is the total number of clients, and \(d\) is the gradient dimension. Additionally, signSGD and RLR can be adapted to the RAR architecture, following the procedure discussed in Section~\ref{sec:alg}, resulting in the same communication cost as our method: \(\frac{d(n-1)(m+1)}{n}\). 
Note that RAR-based signSGD and RLR achieve the same classification performance as their SC-based counterparts, differing only in reduced communication overhead.
% !TEX root = mainfile.tex

\section{Conclusion}
\label{sec:conclusion}

We introduce \algns, the first RAR-based FL algorithm, ensuring robustness and communication efficiency. \alg resolves  server bottlenecks in traditional SC FL systems, with theoretical convergence guarantees under Byzantine settings and strong performance demonstrated through experiments on real-world datasets.

\begin{acks}
This work was supported in part by NSF grants CAREER CNS-2110259, CNS-2112471, IIS-2324052, DARPA YFA D24AP00265, ONR grant N00014-24-1-2729, AFRL grant PGSC-SC-111374-19s, and Cisco Research Award PO-USA000EP312336.
\end{acks}

\bibliographystyle{plain}
\bibliography{refs}

% !TEX root = mainfile.tex

\appendix

\section{Proof of Theorem~\ref{theorem_1}} \label{sec:appendix_1}

The following proof is partially inspired by~\cite{ozdayi2021defending}.
We represent \( I_{\left(\sum_{i \in [n]} \text{sign}\left(\bm{g}_i^t[k]\right) > \lambda\right)} \) as a \( d \times d \) matrix \( \bm{I}_t \), where \( d \) represents the dimension of \( \bm{w}^t \). For all \( k, j \in [d] \), \( \bm{I}_t[k, k] = I_{\left(\sum_{i \in [n]} \text{sign}\left(g_{i,t}^{k}\right) > \lambda\right)} \) and \( \bm{I}_t[k, j] = 0 \) when \( k \neq j \).
This means that our proposed \alg at round \( t \) can be expressed as follows:
\begin{align}
& \mathsf{GAR}(\bm{g}_1^t, \bm{g}_2^t,...,\bm{g}_n^t)[k] = 
\begin{cases}
1, & \text{if } \sum_{i \in [n]} \text{sign}(\bm{g}_i^t[k]) > \lambda, \\
-1, & \text{if } \sum_{i \in [n]} \text{sign}(\bm{g}_i^t[k]) \leq \lambda,
\end{cases}  \\
& = 2 \cdot I_{\left(\sum_{i \in [n]} \text{sign}\left(\bm{g}_i^t[k]\right) > \lambda\right)} - 1
= 2 \cdot \bm{I}_t[k,k] - 1.
\end{align}

Let \( f_{i}(\bm{w}) = \mathbb{E}_{\mathcal{D}_i}f_{i}\left(\bm{w}, \xi_{i}\right) \) represent the loss function of the \( i \)-th client, where \( \xi_{i} \) denotes the randomness introduced by local batch variability, and these random variables are independent of each other. Thus, we have \( \bm{g}_i^t = \nabla f_{i}(\bm{w}_i^t, \xi_i^t) \). 
Additionally, let \( \mathbb{E}_{\mathcal{D}_i} (\bm{g}_i^t \mid \mathcal{F}_t) = \nabla f_i(\bm{w}_i^t) \), where \( \mathcal{F}_t \) represents the filtration generated by all random variables at round \( t \). Specifically, \( \mathcal{F}_t \) is a sequence of increasing \(\sigma\)-algebras \( \mathcal{F}_s \subseteq \mathcal{F}_t \) for all \( s < t \).

Therefore, the update rule can be expressed as:
\begin{align}
\label{update_rule}
\bm{w}^{t+1} = \bm{w}^t - \eta \left(2 \bm{I}_t - \mathbb{I}\right), 
\end{align}
where \(\mathbb{I}\) denotes a \(d \times d\) identity matrix.
By Assumption 1, we obtain:
\begin{align}
\label{Assumption_equation}
f(\bm{w}^{t+1}) \leq f(\bm{w}^{t}) + \left\langle \nabla f(\bm{w}^t), \bm{w}^{t+1} - \bm{w}^{t} \right\rangle + \frac{L}{2} \| \bm{w}^{t+1} - \bm{w}^{t}\|^2.
\end{align}

By combining Eq.~(\ref{update_rule}) and Eq.~(\ref{Assumption_equation}), it follows that:
\begin{align}
\label{update_rule_further}
f(\bm{w}^{t+1}) \leq f(\bm{w}^t) - \eta \left\langle \nabla f(\bm{w}^t), \left(2 \bm{I}_t - \mathbb{I}\right)\right\rangle + \frac{L}{2} \left\|\eta \left(2 \bm{I}_t - \mathbb{I}\right)\right\|^2.
\end{align}

Note that \( \|2\bm{I}_t - \mathbb{I}\| = 1 \).  
Taking the conditional expectation with respect to the filtration \( \mathcal{F}_t \) in Eq.~(\ref{update_rule_further}), we obtain:
\begin{align}
\mathbb{E}_{\mathcal{F}_t}  f(\bm{w}^{t+1}) \leq \mathbb{E} f(\bm{w}^t)- \eta\left\langle \nabla f(\bm{w}^t), \mathbb{E}_{\mathcal{F}_t} \left(2 I_{t} - \mathbb{I}\right) \right\rangle + \frac{L\eta^2}{2}. 
\end{align}

We obtain a convergence result in a scenario where, for a given \( t \), we alter the signs of either none or all elements. Specifically, for each \( k = 1, \ldots, d \), we set \( \bm{I}_t[k, k] = 0 \) or \( \bm{I}_t[k, k] = 1 \), respectively. For ease of notation, we denote these two cases as \( \bm{I}_t = 1 \mid \mathcal{F}_t \) and \( \bm{I}_t = 0 \mid \mathcal{F}_t \).
Let \( p_{\bm{I}_{t}=1 \mid \mathcal{F}_{t}} \) and \( p_{\bm{I}_{t}=0 \mid \mathcal{F}_{t}} \) represent the probabilities that \( \bm{I}_t = 1 \mid \mathcal{F}_t \) and \( \bm{I}_t = 0 \mid \mathcal{F}_t \), respectively.
Thus, it follows that:
\begin{align}
\mathbb{E}_{\mathcal{F}_t} \left( 2 \bm{I}_{t} - \mathbb{I}\right)= - p_{\bm{I}_{t}=0 \mid \mathcal{F}_{t}} +  p_{\bm{I}_{t}=1 \mid \mathcal{F}_{t}}  =(1 - 2 p_{\bm{I}_{t}=0 \mid \mathcal{F}_{t}}).
\end{align}

Therefore, it follows that:
\begin{align}
\label{lala}
\mathbb{E}  f(\bm{w}^{t+1}) \leq \mathbb{E}f(\bm{w}^t)- \eta(1 - 2 p_{\bm{I}_{t}=0 \mid\mathcal{F}_{t}})\| \nabla f(\bm{w}^t) \| + \frac{L\eta^2}{2}.
\end{align}

Now, considering that \( \Pr(\bm{I}_t[1,1] = \dots = \bm{I}_t[d,d] = 0 \mid \mathcal{F}_t) \leq \Pr(\bm{I}_t[i,i] = 0 \mid \mathcal{F}_t) < 0.5 \), we can apply telescoping and rearrange Eq.~(\ref{lala}) to derive the following:
\begin{align}
\frac{1}{T} \sum_{t=1}^{T} \|\nabla f(\bm{w}^t)\| \leq \frac{f(\bm{w}^1) - f^*}{\eta T } + \frac{L\eta^2}{2}.
\end{align}

Thus, we complete the proof.

\end{document}